\documentclass[conference,a4paper]{IEEEtran}

\pdfoutput=1
\usepackage[utf8]{inputenc} 
\usepackage[T1]{fontenc}
\usepackage{url}
\usepackage{ifthen}
\usepackage{cite}

\usepackage[cmex10, tbtags]{amsmath}
\usepackage[pdftex]{graphicx}
\usepackage{hhline}
\usepackage[table]{xcolor}
\usepackage{amssymb}
\usepackage[english]{babel}
\usepackage[normalem]{ulem}
\usepackage{epsfig,epsf,subfigure,graphicx,graphics}
\usepackage{tikz}
\usepackage{multirow}
\usepackage{cmap}
\usepackage{enumerate}
\usepackage{xcolor}
\usepackage{soul}
\usepackage{arydshln}
\usepackage{amsthm}

\newtheorem{theorem}{Theorem}
\newtheorem{lemma}{Lemma}
\newtheorem{remark}{Remark}
\graphicspath{{./fig/}}

\renewcommand\hl[1]{#1}

\def \IncludeAppendices{True}

\begin{document}
\title{Degrees of Freedom of the Bursty MIMO X Channel with Instantaneous Topological Information} 

\author{
\IEEEauthorblockN{
Shih-Yi Yeh\IEEEauthorrefmark{1}\IEEEauthorrefmark{2} and I-Hsiang Wang\IEEEauthorrefmark{1}\IEEEauthorrefmark{3}
}
\IEEEauthorblockA{
\IEEEauthorrefmark{1}Graduate Institute of Communication Engineering, 
\IEEEauthorrefmark{3}Department of Electrical Engineering\\
National Taiwan University, Taipei, Taiwan\\
\IEEEauthorrefmark{2}Email: \url{steven0416@gmail.com}
\IEEEauthorrefmark{3}Email: \url{ihwang@ntu.edu.tw}
}
}

\maketitle

\begin{abstract}
We study the effects of instantaneous feedback of channel topology on the degrees of freedom (DoF) of the bursty MIMO X channel, where the four transmitter-receiver links are intermittently on-and-off, governed by four independent Bernoulli $(p)$ random sequences, and each transmitter and receiver are equipped with $M$ and $N$ antennas, respectively.  We partially characterize this channel:  The sum DoF is characterized when $p\le \frac{1}{2}$ or when $\frac{\min(M,N)}{\max(M,N)} \le \frac{2}{3}$.  In the remaining regime, the lower bound is within $5.2\%$ of the upper bound.  Strictly higher DoF is achieved by coding across channel topologies.  In particular, codes over as many as $5$ topologies are proposed to achieve the sum DoF of the channel when $p\le \frac{1}{2}$.  A transfer function view of the network is employed to simplify the code design and to elucidate the fact that these are space-time codes, obtained by interference alignment over space and time.

\end{abstract}

\ifdefined\IncludeAppendices
\else
\textit{A full version of this paper is accessible at:\\} \url{http://homepage.ntu.edu.tw/%7Eihwang/Eprint/isit18bx.pdf}
\fi

\section{Introduction}


Interference is a critical issue in wireless communication, limiting the capacity of a  network, and the two-user-pair interference channel (IC) has been a canonical model for studying the capacity of interference networks.  The degrees of freedom (DoF) of an MIMO IC with three antennas at each terminal, for instance, is only $3$, instead of $6$ when there is no interference between the two pairs of users \cite{JF07}.  An interesting discovery is made in \cite{MMK06}--{\hspace{1sp}}\cite{MMK08}, however, that the sum DoF of this network can be easily increased to $4$ by simply allowing cross messaging between the two pairs of users, and this network is referred to as X channel (XC) in the literature.  Recently, however, questions were raised as to how the capacity of an interference network would change when the links between the transmitters and receivers exist only intermittently due to frequency hopping, shadowing, co-channel interference, ... etc, c.f. \cite{WSDV13}, \cite{YW18a}, \cite{VMA14} among others.  This conceptually simple change turns out to have profound implications.  Take the bursty MIMO XC for example \cite{YW18a}. Burstiness of the channel disrupts the network topology, turning the XC into a \emph{new} network with $16$ different topologies.  This significantly changes its channel capacity and achievability schemes, and greatly complicates the characterization of the sum DoF.

Availability of channel state information at the transmitters (CSIT) is long known to have a great impact on the channel capacity.  Interestingly, it is also discovered in \cite{MAT12} that delayed CSIT is still very useful, even if it is completely stale.  This motivates a sequence of works to further explore the benefits of delayed CSIT, including \cite{GMK11}--{\hspace{1sp}}\cite{KA17} for the IC and XC.  Moreover, for networks with time-varying topology, communication rate gains have been reported even with only topological information at the transmitters \cite{SGJ13}, \cite{LKA16}.  The highlight of the achievability schemes in these works is coding over multiple channel uses or topologies.  This motivates us to consider how channel topology information at the transmitters (CTIT) may be used to enhance the achievable rates on the bursty MIMO XC.  As a first step, we consider instantaneous feedback of channel topology to the transmitters in this work.

Unlike \cite{GMK11}--{\hspace{1sp}}\cite{LKA16}, where the channel matrices are time-varying, we study the bursty MIMO XC whose channel matrices are drawn from a continuous distribution and are fixed throughout the communication.  The only time-varying component in this channel is its topology, which is assumed known to the receivers and is fed back to transmitters instantaneously.  Each transmitter and each receiver are equipped with $M$ and $N$ antennas, respectively.  The four links between the transmitters and receivers are on-and-off intermittently, governed by four independent Bernoulli $(p)$ random sequences, similar to the model in \cite{VMA14}.  For this bursty MIMO XC, we ask these questions:  How may we exploit the topology feedback to achieve higher DoF?  What is its sum DoF?  How does it compare to the case where there is no topology feedback \cite{YW18a} or no cross-link messaging \cite{VMA14}?

Our key findings are the following:  First, strictly higher DoF can be achieved on this bursty MIMO XC by coding across channel topologies.  In particular, sophisticated codes across as many as $5$ topologies prove beneficial on this channel.  This in contrast to the simpler codes for the channels considered in \cite{LKA16}, \cite{SGJ13}, or \cite{VMA14}.  Secondly, the search of DoF-optimal codes by trials-and-errors is prohibitive due to the large space of coding possibilities for this channel.  The \emph{transfer function view} of the \emph{parallel} channel across topologies, on the other hand, affords a systematic approach that dramatically reduces the effort of code design and makes it much more manageable.  It also elucidates the fact that these are space-time codes, obtained by interference alignment over space and time.   A similar observation of the interference alignment interpretation is also made in \cite{MAT12} albeit for the broadcast channel. Thirdly, armed with these codes, we give a partial characterization of the sum DoF of this channel.  The sum DoF is determined when $p\le \frac{1}{2}$, or when the antenna ratio $r$, defined as $\frac{\min(M,N)}{\max(M,N)}$, is no greater than $\frac{2}{3}$.  When $r > \frac{2}{3}$ and $p > \frac{1}{2}$, the sum DoF is not fully characterized.  However, we provide a lower bound that is within $5.2\%$ of the upper bound, in the worst case.  Figure \ref{fig:DoF_Compare} illustrates the sum DoF of this channel, the benefits of coding across topologies, and how the sum DoF of the channel varies when topology feedback or cross-link messaging is not allowed.


\begin{figure}[t]
\centering
\includegraphics[width=7cm]{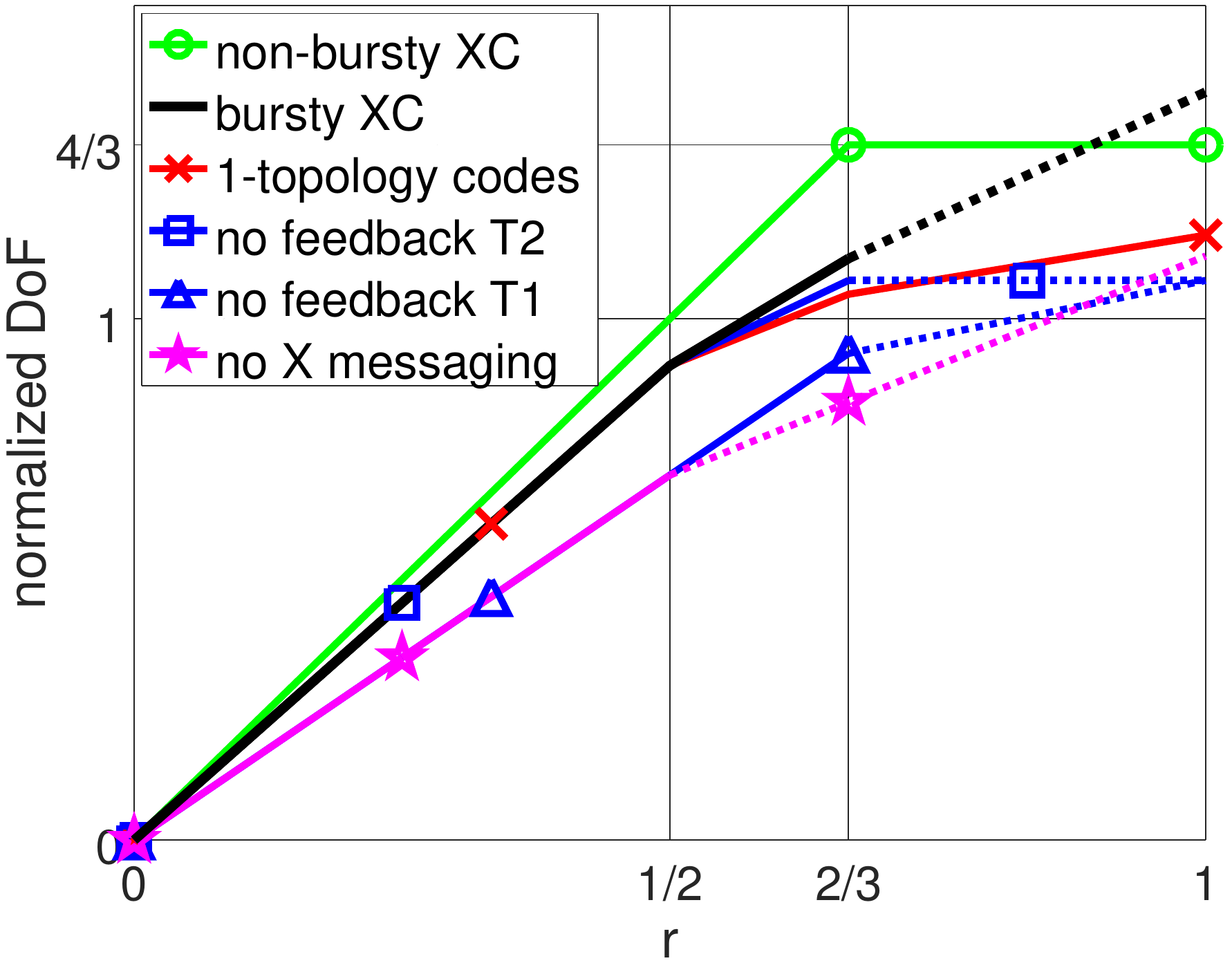}
\caption{The normalized sum DoF of the channel when $p=0.7$ (The DoF normalization is against $\max(M,N)$.  1-topology codes: No coding across topologies.  T2: $M\ge N$. T1: $M \le N$.  Dashed lines: lower bounds.)}
\label{fig:DoF_Compare}
\end{figure}

\section{Problem Formulation} \label{Prob}


The system model of the bursty MIMO XC is depicted in Figure \ref{fig:BXC_Model}.  There are two transmitters and two receivers in the system, denoted by Tx$i$ and Rx$j$, respectively, for $i, j\in\{1, 2\}$.  Each transmitter is equipped with $M$ antennas, while each receiver has $N$ antennas.  $M_{ji}\sim\mathrm{Unif}\{1,2,\ldots,2^{nR_{ji}}\}$ denotes the message from Tx$i$ to Rx$j$, \hl{encoded over a code block of $n$ symbols with code rate $R_{ji}$}, and $\hat M_{ji}$ is the decoded message at Rx$j$.  $X_i$ represents the signal transmitted by Tx$i$ and $Y_j$ is the received signal at Rx$j$.  Each transmitter has an average transmit power constraint $P$, i.e. $\frac{1}{n}\sum_{k=1}^n{\|X_i[k]\|^2} \le P, i\in \{1, 2\}$, where $X_i[k]$ denotes the $k$-th transmitted symbol of Tx$i$.  $H_{ji}$ models the $N\times M$ channel matrix from Tx$i$ to Rx$j$.  To simplify the notations in Section \ref{Achi} and \ref{Tran}, we assign the following aliases: $H_1=H_{11}, H_2=H_{12}, H_3=H_{21}$ and $H_4=H_{22}$. The channel matrices are drawn randomly from a continuous distribution with i.i.d. elements, but are fixed during the transmission.  Each transmitter or receiver is assumed to have perfect knowledge of all channel matrices.  $Z_j$ is the additive Gaussian noise at Rx$j$ with zero mean and unit variance, i.i.d. in time.

The four Tx-Rx links are intermittently on and off, controlled by four independent and identically distributed Bernoulli $(p)$ random sequences, $S_{11}[k]$, $S_{12}[k]$, $S_{21}[k]$, and $S_{22}[k]$.  The link from Tx$i$ to Rx$j$ is on with probability $p$ at the $k$-th time instant when $S_{ji}[k]=1.$  The link is off if $S_{ji}[k]=0$.  For convenience we define $q\triangleq 1-p$, and for brevity of notation, we may drop the dummy time index ($k$) hereafter and abbreviate $S_{ji}[k]$ as $S_{ji}$ when there is no confusion.  Each receiver has perfect knowledge of the burstiness of the two incoming links, e.g. Rx$1$ knows $S_{11}$ and $S_{12}$, and feeds this topological information back to both transmitters instantaneously.

A rate tuple $(R_{11}, R_{12}, R_{21}, R_{22})$ is said to be achievable on the bursty MIMO X channel if there exists a sequence of codes such that $\mathcal P\{\hat M_{ji} \ne M_{ji}, \textrm{ for some } i, j\in\{1,2\}\}$ converges to zero as the block length of the codes tends to infinity.  The capacity region of the channel is the set of all achievable rate tuples $(R_{11}, R_{12}, R_{21}, R_{22})$\hl{, and the sum capacity of the channel, $C_\mathrm{sum}$, is the supremum of the achievable sum rates $(R_{11} + R_{12}+ R_{21}+ R_{22})$}.  The sum DoF, $\eta$, of the channel follows conventional definition, i.e.
\begin{equation} \label{eq:2_1}
\eta \triangleq \lim_{P\to\infty} \frac{C_\mathrm{sum}}{(\frac{1}{2})\log(P)}.
\end{equation}
For convenience, we also define the \emph{normalized} sum DoF to be $\eta/\max(M,N).$

In this paper we evaluate the sum DoF of the channel in the almost surely (a.s.) sense, since the channel matrices are drawn from a continuous probability distribution as in \cite{JS08}.

\begin{figure}[t]
\centering
\includegraphics[width=8cm]{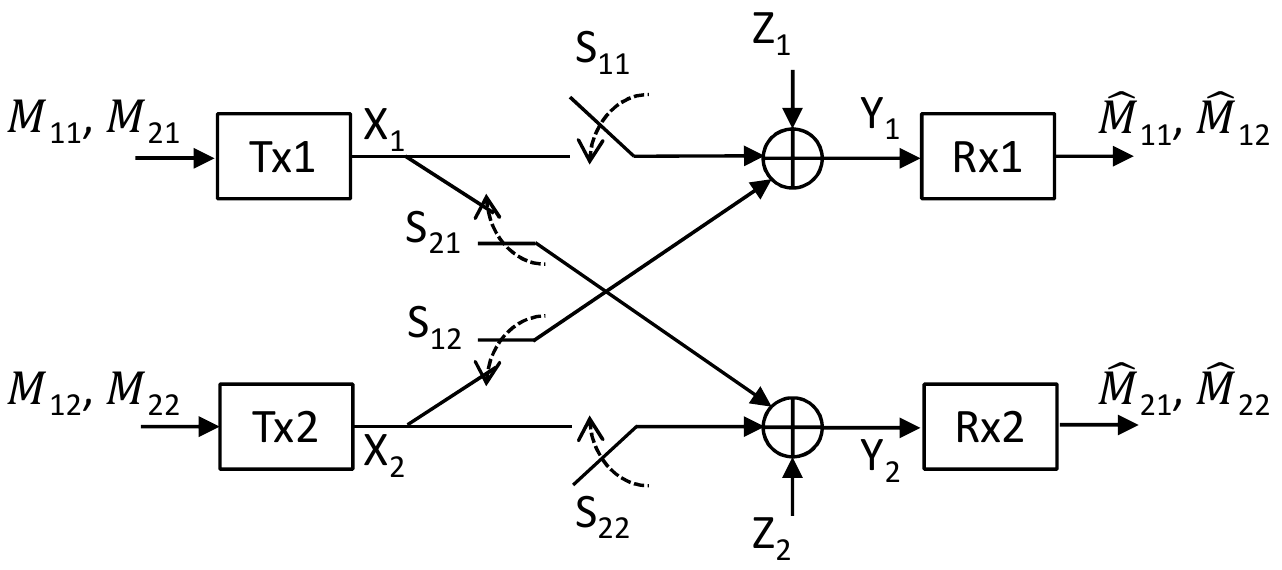}
\caption{Model of the bursty MIMO X channel}
\label{fig:BXC_Model}
\end{figure}

\section{Main Results} \label{Main}


The normalized sum DoF of the bursty MIMO XC with instantaneous feedback of channel topology is characterized and bounded by the following two theorems.

\begin{theorem} \label{thm:DoF}
The normalized sum DoF of this channel, when $r\le \frac{2}{3}$ or $p \le \frac{1}{2}$, is given by
\begin{equation*}
\left\{ \begin{aligned}[rl]
2rp(1+q), & \quad r \le \frac{1}{2} \\
2r(p^2+2pq^2)+2p^2q,& \quad r> \frac{1}{2}\textrm{ and } (r\le \frac{2}{3}\textrm{ or } p\le \frac{1}{2}).\\
\end{aligned} \right.
\end{equation*}
\end{theorem}
(Recall that $r\triangleq \frac{\min(M,N)}{\max(M,N)}$ and $q\triangleq 1-p$.)

\begin{theorem} \label{thm:UBLB}
When $r> \frac{2}{3}$ and $p>\frac{1}{2}$, the normalized sum DoF is upper bounded by $\min(\eta^\mathrm{ub}_1, \eta^\mathrm{ub}_2)$, where
\begin{align*}
\eta^\mathrm{ub}_1 &\triangleq 2r(p^2+2pq^2)+2p^2q,  \\
\eta^\mathrm{ub}_2 &\triangleq 4rpq + \frac{4}{3}p^2,
\end{align*}
and is lower bounded by $\eta^\mathrm{lb}$, given by
\begin{equation*}
\eta^\mathrm{lb} \triangleq rpq(4q^2+6p)+2p^2q+\frac{4}{3}(p^4-p^3q).
\end{equation*}
\end{theorem}

\begin{remark}
It is easily verified that $\eta_\mathrm{lb}$ is within $5.2\%$ of $\min(\eta^\mathrm{ub}_1, \eta^\mathrm{ub}_2)$, and the maximum gap occurs  when $r\simeq 0.81$ and $p\simeq 0.77$.
\end{remark}

The sum DoF of this bursty MIMO XC has the following properties, as illustrated in Figure \ref{fig:DoF_Compare}:
\begin{enumerate}
\item The sum DoF of the bursty MIMO XC can be larger than that of the non-bursty channel when $r$ is large, e.g. $r\simeq 1$.  In contrast, without transmitter knowledge of channel topology (CTIT), the best known achievable sum DoF of the bursty channel is always lower.
\item However, when $r\le \frac{2}{3},$ burstiness of the channel, i.e. $p<1$, always reduces the sum DoF of the channel. 
\item Coding across channel topologies can lead to strictly higher sum DoF, but only when $r > \frac{1}{2}$.
\item When $r\le \frac{1}{2}$ and $M\ge N$, lack of CTIT does not decrease the sum DoF.
\item When $r\le \frac{1}{2}$ and $M\le N$, lack of CTIT and lack of cross messaging both lead to the same lower sum DoF.
\item Existence of cross-links can increase the sum DoF when the channel is bursty ($p<1$).
\end{enumerate}

\ifdefined\IncludeAppendices
We prove the achievability of Theorem \ref{thm:DoF} and $\eta_\mathrm{lb}$ of Theorem \ref{thm:UBLB} for $M\ge N$ in the next section.  The rest of the proof can be found in the appendices, including the converse proof and the $M \le N$ case.
\else
Due to the limitation of space, we prove the achievability of Theorem \ref{thm:DoF} and $\eta_\mathrm{lb}$ of Theorem \ref{thm:UBLB} for $M\ge N$ in this paper.  The rest of the proof can be found in the full version of this paper, including the converse proof and the $M \le N$ case.
\fi

\section{Achievability Schemes and DoF Lower Bounds} \label{Achi}


\begin{figure}[t]
\centering
\includegraphics[width=9cm]{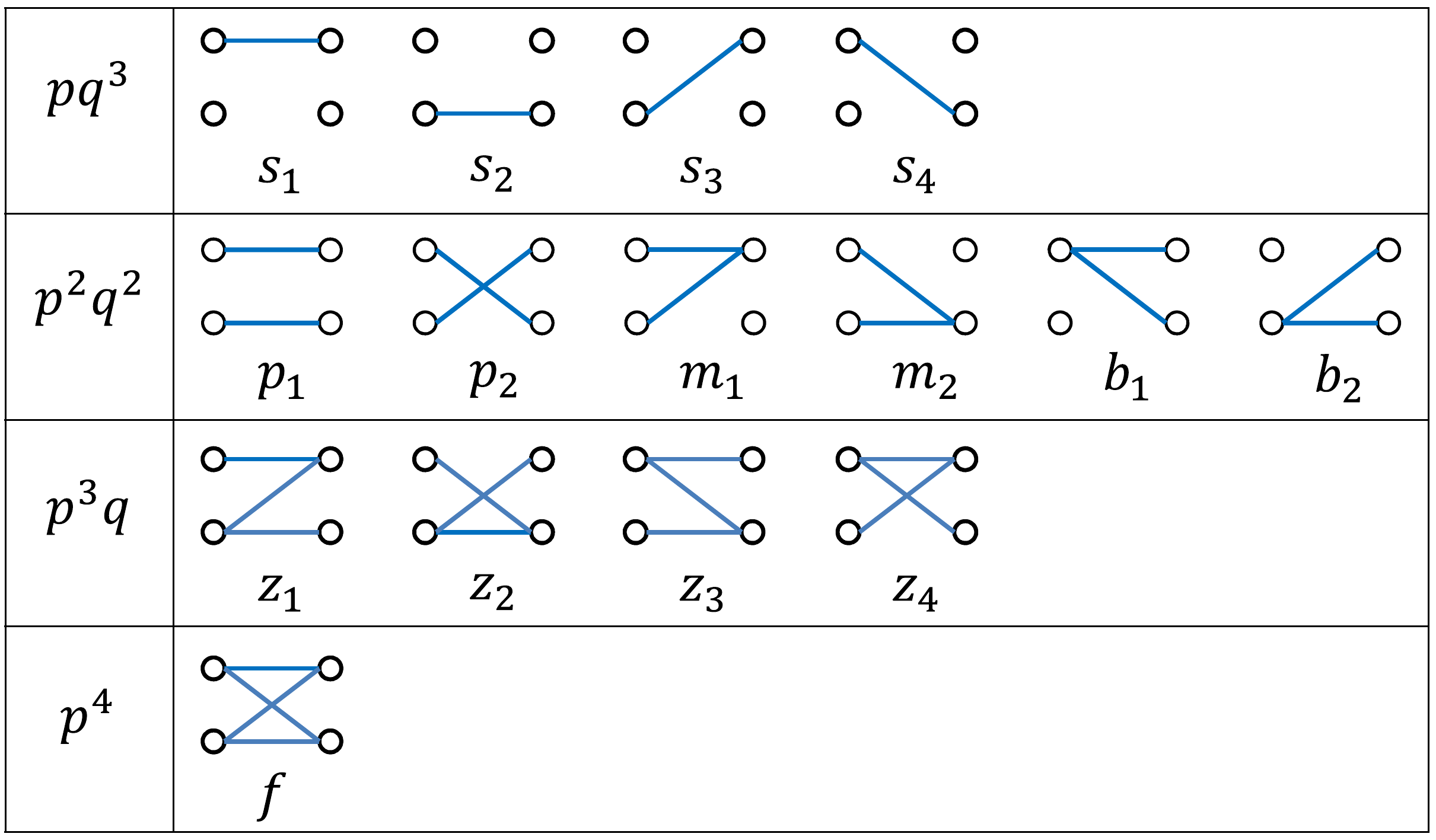}
\caption{Topologies of the Bursty MIMO XC (Topologies on the same row have the same probability, indicated in the leftmost column.)}
\label{fig:BXC_States}
\end{figure}

In this section, we present the coding schemes and prove achievability of the sum DoF given by Theorem \ref{thm:DoF} and the $\eta^\mathrm{lb}$ lower bound of Theorem \ref{thm:UBLB} with $M\ge N$.  Key to the proof are the following two lemmas which establish the sum DoF of two parallel MIMO channels, each consisting of a subset of the topologies illustrated in Figure \ref{fig:BXC_States}.

\begin{lemma} \label{lemma1}
For the $\{z_1, z_2\}$ parallel MIMO channel consisting of the $z_1$ and $z_2$ topologies, $2N+M$ sum DoF is achievable (a.s.), when $\frac{1}{2} <\frac{N}{M} \le 1$.
\end{lemma}
\begin{proof}
To prove this, we combine the strategy of coding across topologies in \cite{LKA16} with interference nulling beamforming in \cite{YW18a}, \cite{MMK08}.  As illustrated in Figure \ref{fig:z1z2_code}, the $\phi_2$ and $\phi_4$ filters are $M\times (M-N)$ full-rank matrices satisfying $H_2\phi_2=H_4\phi_4=0,$ and the $\tilde I_M$ and $\hat I_M$ consist of the first $N$ and $2M-N$ columns of the identity matrix $I_M$, respectively.  $a, b, c, d, e$ denote vectors of $N, M-N, 2N-M, N, M-N$ real variables, respectively.  

When the signal power ($P$) is large, it is obvious from the schematic that $(b,c)$ can be solved reliably at Rx$2$ and so can $(e,c)$ at Rx$1$.  Moreover, since $H_{2}\phi_2=0$, Rx$1$ also receives a linear combination of $a$ and $c$, denoted by $L(a,c)$, plus noise, from which $a$ can be solved reliably as $c$ is known.  By the same token, $d$ can also be solved reliably at Rx$2$. Hence we can communicate a total of $2N+M$ variables reliably as $P$ tends to infinity, proving the achievability of the sum DoF.
\end{proof}

\begin{figure}[t]
\centering
\includegraphics[width=9cm]{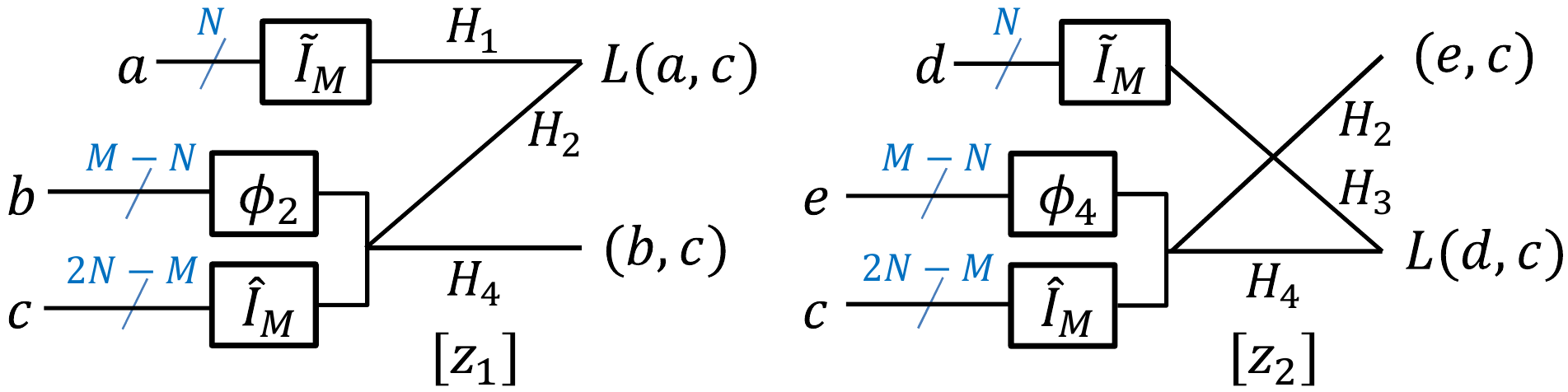}
\caption{Coding scheme for the $\{z_1, z_2\}$ parallel MIMO channel($\tilde I_M, \hat I_M:$ first $N$ and $2M-N$ columns of $I_M$, respectively.  $H_2\psi_2=H_4\phi_4=0.$ )}
\label{fig:z1z2_code}
\end{figure}

\begin{figure}[b]
\centering
\includegraphics[width=9cm]{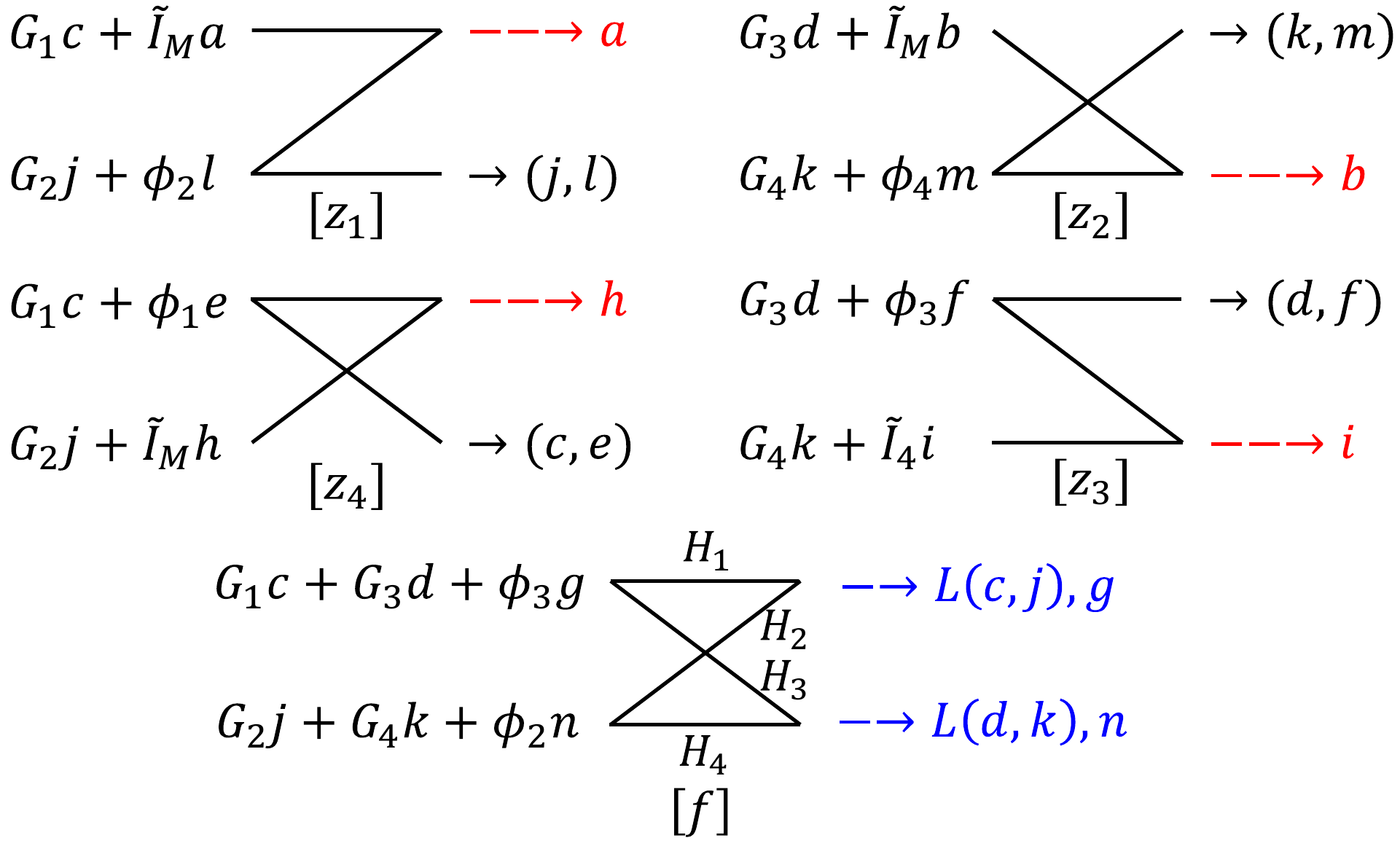}
\caption{Coding scheme for the $\{z, f\}$ parallel MIMO channel}
\label{fig:zf_code}
\end{figure}

\begin{table}[b]
\renewcommand{\arraystretch}{1.5}
\centering
\caption{Length of vectors in the $\{z,f\}$ coding scheme}
\label{tab:zf_vectors}
\vspace{-2ex}
\begin{tabular}{|c|c|c|c|}
	\hline
	  \bfseries vector length & $N$ & $2N-M$ & $M-N$ \\ 
	\hline 
	& & & \\ [-3.5ex]
	\hline
	\bfseries Tx1 vectors & $a,b$ & $c,d$ & $e,f,g$ \\
	\hline
	\bfseries Tx2 vectors & $h,i$ & $j,k$ & $l,m,n$ \\
	\hline
\end{tabular}
\end{table}

Since the $\{z_3, z_4\}$ parallel channel is identical to the $\{z_1, z_2\}$ channel after a relabeling, it obviously has the same DoF.  Hence $2(2N+M)+\frac{4}{3}M$ sum DoF is achievable on the parallel channel comprising the $z_1, z_2, z_3, z_4, f$ topologies, when $\frac{1}{2} <\frac{N}{M} \le 1$.  Interestingly, we can do better.

\begin{lemma} \label{lemma2}
For the $\{z, f\}$ parallel MIMO channel consisting of the $z_1, z_2, z_3, z_4$ and $f$ topologies, $6N+2M$ sum DoF is achievable (a.s.), when $\frac{2}{3} <\frac{N}{M} \le 1$.
\end{lemma}
\begin{proof}
The key idea is to incorporate both interference alignment \cite{JS08} and interference nulling \cite{YW18a},\cite{MMK08} into coding across topologies.  As in the proof of Lemma \ref{lemma1}, let $\tilde I_M$ consist of the first $N$ columns of $I_M$, and let $\phi_i$ be an $M\times (M-N)$ full-rank matrix satisfying $H_i\phi_i=0, i=1,2,3,4$.  In addition, let us define $G_i$ to comprise the first $(2N-M)$ columns of the pseudo inverse of $H_i$, namely $H_i^T(H_iH_i^T)^{-1}$, and consider the coding scheme depicted in Figure \ref{fig:zf_code}, where $a,b,c,...,n$ are vectors of real variables with their length specified in Table \ref{tab:zf_vectors}.

Assuming large signal power ($P$), we decode with successive interference cancellation in three steps:  

Step 1: Decode the variables at the receiver of each $z_i$-topology with only one incoming link (indicated by an $\rightarrow$ in the figure).  At Rx$2$ of the $z_1$-topology, for example, $(j,l)$ can clearly be decoded reliably.  The other $z_i$-topologies can be treated similarly.

Step 2: Decode the variables at both receivers of the $f$-topology.  Consider Rx$1$ first.  Since $(d,k)$ have been decoded in Step 1, we can remove them.  Note also that vector $n$ is gone due to interference nulling, i.e. $H_2\phi_2=0$.  Moreover, vectors $c$ and $j$ are aligned because $H_1G_1=H_2G_2$.  As a result, we can reliably decode $g$ and $L(c,j)$, a linear combination of $c$ and $j$, as illustrated in Figure \ref{fig:zf_code}.  Rx$2$ is decoded similarly.  

Step 3: Finally, we decode the remaining receiver of each $z_i$-topology.  Take Rx$1$ of the $z_1$-topology for instance.  Vector $l$ is nulled, while $c$ and $j$ are aligned and $L(c,j)$ has been decoded in Step 2.  Canceling it, vector $a$ can hence be decoded reliably.  The same strategy applies to the other $z_i$-topologies.  

Therefore, vectors $a,b,c,...,n$ can all be reliably decoded, leading to an achievable sum DoF of $\setlength{\medmuskip}{2mu} 4N+4(2N-M)+6(M-N)=6N+2M.$
\end{proof}

\begin{remark}
With simple converse arguments, one can show that the DoF achieved in Lemma \ref{lemma1} and \ref{lemma2} are in fact optimal.
\end{remark}

\subsection{Achievability proof of Theorem \ref{thm:DoF}}

To prove Theorem \ref{thm:DoF}, we distinguish three cases:  

When $N\le \frac{1}{2}M$, it is unnecessary to code across topology.  With DoF-optimal code for each topology, it is easy to verify that the following sum DoF is achievable (a.s.):
\begin{equation} \label{eq:4_1}
N(4pq^3+10p^2q^2+8p^3q+2p^4)=2Np(1+q).
\end{equation}

When $\frac{1}{2}M < N\le \frac{2}{3}M$, we use codes across $\{z_1, z_2\}$ topologies and $\{z_3, z_4\}$ topologies, together with per-topology DoF-optimal codes for the remaining topologies.  With Lemma \ref{lemma1}, it follows that we can achieve (a.s.) a sum DoF of
\begin{align} \begin{split} \label{eq:4_2}
&pq^34N+p^2q^2(6N+2M)+p^3q(4N+2M)+p^42N\\
=&2N(p^2+2pq^2)+2Mp^2q.
\end{split} \end{align}

Lastly, let us consider the case where $\frac{2}{3}M < N\le M$ and $p\le \frac{1}{2}$. For a long period of ($n$) channel uses, the $f$-topology occurs approximately $np^4$ times, while each  $z_i$-topology occurs approximately $np^3q$ times.  We first code across the $\{z_1, z_2, z_3, z_4, f\}$ topologies and totally consume the $f$-topologies.  Since $p^4 \le p^3q$, we then use $\{z_1, z_2\}$- and $\{z_3, z_4\}$-topological codes on the remaining $z_i$-topologies.  For the other topologies, simply employ a DoF-optimal code on each topology.  Thus, by Lemma \ref{lemma1} and \ref{lemma2}, we can achieve sum DoF of $\setlength{\medmuskip}{2mu} pq^34N+p^2q^2(6N+2M)+p^4(6N+2M)+(p^3q-p^4)(4N+2M)$, or equivalently
\begin{equation}  \label{eq:4_3}
2N(p^2+2pq^2)+2Mp^2q,
\end{equation}
which, interestingly, coincides with (\ref{eq:4_2}).  The achievability of Theorem \ref{thm:DoF} is hence established by (\ref{eq:4_1})--(\ref{eq:4_3}).

\subsection{Proof of $\eta^\mathrm{lb}$ of Theorem \ref{thm:UBLB}}

When $\frac{2}{3}M < N\le M$ and $p> \frac{1}{2}$, the priority is again to use the $\{z, f\}$-topological code as much as possible.  For the remaining topologies, DoF-optimal code is employed on each of them.  Noting that $p^4>p^3q$, we conclude that the following sum DoF is achievable (a.s.): $\setlength{\medmuskip}{2mu} pq^34N+p^2q^2(6N+2M)+p^3q(6N+2M)+(p^4-p^3q)\frac{4}{3}M$, or equivalently
\begin{equation} \label{eq:4_4}
pq^34N+p^2q(6N+2M)+(p^4-p^3q)\frac{4}{3}M,
\end{equation}
proving the $\eta^\mathrm{lb}$ of Theorem \ref{thm:UBLB}.

\section{Transfer Function View} \label{Tran}


For simple codes across a couple of topologies, the DoF-optimal codes are not hard to find by inspecting the schematic, e.g. Figure \ref{fig:z1z2_code}, and trials and errors.  This approach, however, quickly becomes impractical as the topologies and antennas increase.  Take the $5$-topology parallel channel shown in Figure \ref{fig:zf_code}, for example.  This parallel channel has $10$ transmitters and $10$ receivers, and with $M=4$ and $N=3$, its sum DoF is $26$.  To find a DoF-optimal code from the schematic by trials-and-errors, we need to decide how to distribute these $26$ variables among the $10$ transmitters.  Some variables may be used multiple times and combined with other variables.  For each variable, we also have the flexibility of choosing a beamforming vector.  There are simply too many possibilities---assuming even just $4$ coding choices on each transmitter, this would amount to $4^{10}$ possibilities!  Not to mention the decoding schemes across the receivers.  We need a more systematic method.

One such approach can be obtained from the transfer function view of the \emph{entire} $\{z,f\}$-parallel channel, namely
\begin{equation*}
\setlength{\dashlinedash}{.4pt}	\setlength{\dashlinegap}{.8pt}
\setlength{\medmuskip}{0mu} 	\setlength{\thickmuskip}{0mu}
\left[ \arraycolsep=1.2pt \begin{array}{c}
Y_{1,z_1}\\
Y_{1,z_2}\\
Y_{1,z_3}\\
Y_{1,z_4}\\
Y_{1,f} \vspace{1pt} \\
\hdashline  \\ [-1.0em]
Y_{2,z_1}\\
Y_{2,z_2}\\
Y_{2,z_3}\\
Y_{2,z_4}\\
Y_{2,f}\\
\end{array}\right]
=\left[ \arraycolsep=1.2pt\def\arraystretch{1.0}
\begin{array}{ccccc:ccccc}
H_1 & 0 & 0 & 0 & 0   & H_2 & 0 & 0 & 0 & 0\\
    0 & 0 & 0 & 0 & 0   & 0 & H_2 & 0 & 0 & 0\\
0 & 0 & H_1 & 0 & 0   & 0 & 0     & 0 & 0 & 0\\
0 & 0 & 0 & H_1 & 0   & 0 & 0 & 0 & H_2 & 0\\
0 & 0 & 0 & 0 & H_1   & 0 & 0 & 0 & 0 & H_2 \vspace{1pt} \\ 
\hdashline &&&&&&& \\ [-1.0em]
    0 & 0 & 0 & 0 & 0   & H_4 & 0 & 0 & 0 & 0\\
0 & H_3 & 0 & 0 & 0   & 0 & H_4 & 0 & 0 & 0\\
0 & 0 & H_3 & 0 & 0   & 0 & 0 & H_4 & 0 & 0\\
0 & 0 & 0 & H_3 & 0   & 0 & 0 & 0     & 0 & 0\\
0 & 0 & 0 & 0 & H_3   & 0 & 0 & 0 & 0 & H_4\\
\end{array} \right]
\left[ \arraycolsep=1.2pt \begin{array}{c}
X_{1,z_1}\\
X_{1,z_2}\\
X_{1,z_3}\\
X_{1,z_4}\\
X_{1,f} \vspace{1pt} \\
\hdashline \\ [-1.0em]
X_{2,z_1}\\
X_{2,z_2}\\
X_{2,z_3}\\
X_{2,z_4}\\
X_{2,f}\\
\end{array} \right],
\end{equation*}
where the noise is ignored, and $X_{i,t}, Y_{j,t}$ denote the transmitted and received vector at Tx$i$ and Rx$j$ of the $t$ topology, respectively.  More compactly, we write
\begin{equation} \label{eq:5_2}
\left[ \arraycolsep=1.2pt \begin{array}{c}
\overline{Y_1}\\
\overline{Y_2}\\
\end{array}\right]
=\left[ \arraycolsep=1.2pt\def\arraystretch{1.0}
\begin{array}{cc}
\overline{H_1} & \overline{H_2}\\
\overline{H_3} & \overline{H_3}\\
\end{array} \right]
\left[ \arraycolsep=1.2pt \begin{array}{c}
\overline{X_1}\\
\overline{X_2}\\
\end{array} \right],
\end{equation}
where $\overline{X_i}, \overline{Y_i}$ refer to the \emph{super} vector across $5$ topologies at Tx$i$ and Rx$i$, respectively, and $\overline{H_i}$ denotes the corresponding $5N\times 5M$ \emph{super} channel matrix.  Our goal is to design a precoding matrix $\mathbf P=\mathrm{diag}(A, B)$ so that we can solve the desired number of variables from the transformed system of linear equations:
\begin{equation} \label{eq:5_3}
\left[ \arraycolsep=1.2pt \begin{array}{c}
\overline{Y_1}\\
\overline{Y_2}\\
\end{array}\right]
=\left(
\left[ \arraycolsep=1.2pt\def\arraystretch{1.0}
\begin{array}{cc}
\overline{H_1} & \overline{H_2}\\
\overline{H_3} & \overline{H_3}\\
\end{array} \right]
\left[ \arraycolsep=1.2pt\def\arraystretch{1.0}
\begin{array}{cc}
A & 0\\
0 & B\\
\end{array} \right]
\right)
\left[ \arraycolsep=1.2pt \begin{array}{c}
U_1\\
U_2\\
\end{array} \right]
\triangleq \mathbf H_\mathrm{eff}\left[ \arraycolsep=1.2pt \begin{array}{c}
U_1\\
U_2\\
\end{array} \right].
\end{equation}
Note that $\overline X_1=AU_1, \overline X_2=BU_2,$ where $U_i$ is the effective super input vector at Tx$i$, across topologies. For concreteness, we illustrate the approach with $M=4, N=3$ again.  The extension to general $M$ and $N$ is straightforward.

\subsection{Block-level interference alignment: $24$ DoF achievable}

A moment of reflection on (\ref{eq:5_3}) and the block structure of the $\overline H_i$ super channel matrix suggests the following simple precoding scheme via \emph{block-level} interference alignment:
\begin{equation} \label{eq:5_4}
A = \left[ \arraycolsep=3pt\def\arraystretch{1.0}
\begin{array}{cccc}
0             	& H_1^\dagger 	& \tilde I_4 	& 0 \\
\tilde I_4	& 0 				& 0 			& H_3^\dagger \\
0             	& 0 				& 0 			& H_3^\dagger \\
0             	& H_1^\dagger	& 0 			& 0 \\
0             	& H_1^\dagger	& 0 			& H_3^\dagger  \\
\end{array} \right], \quad
B = \left[ \arraycolsep=3pt\def\arraystretch{1.0}
\begin{array}{cccc}
0             	& H_2^\dagger 	& 0		 	& 0 \\
0			& 0 				& 0 			& H_4^\dagger \\
\tilde I_4  	& 0 				& 0 			& H_4^\dagger \\
0             	& H_2^\dagger	& \tilde I_4	& 0 \\
0             	& H_2^\dagger	& 0 			& H_4^\dagger  \\
\end{array} \right],
\end{equation}
where $\tilde I_4$ consists of the first $3$ columns of identity matrix $I_4$ and $H_i^\dagger$ is the pseudo inverse of $H_i$.  This leads to the following $\mathbf H_\mathrm{eff}$:
\begin{equation} \label{eq:5_5}
\setlength{\dashlinedash}{.4pt}	\setlength{\dashlinegap}{.8pt}
\left[ \arraycolsep=2pt\def\arraystretch{0.9}
\begin{array}{cccc:cccc}
0             	& I_3 	& \tilde H_1 	& 0					& 0	& I_3	& 0			& 0	\\
0			& 0 	& 0 			& 0					& 0	& 0	& 0			& H_2H_4^\dagger	\\
0             	& 0 	& 0 			& H_1H_3^\dagger	& 0	& 0	& 0			& 0	\\
0             	& I_3	& 0 			& 0					& 0	& I_3	& \tilde H_2	& 0	\\
0             	& I_3	& 0 			& H_1H_3^\dagger	& 0	& I_3	& 0			& H_2H_4^\dagger \vspace{1pt}	\\
\hdashline &&&&&&& \\ [-0.9em]
0             	& 0 					& 0	& 0	& 0			& H_4H_2^\dagger	& 0	& 0	\\
\tilde H_3	& 0 					& 0	& I_3	& 0			& 0					& 0	& I_3	\\
0             	& 0 					& 0 	& I_3	& \tilde H_4	& 0					& 0	& I_3	\\
0             	& H_3H_1^\dagger	& 0 	& 0	& 0			& 0					& 0	& 0	\\
0             	& H_3H_1^\dagger	& 0 	& I_3	& 0			& H_4H_2^\dagger	& 0	& I_3	\\
\end{array} \right],
\end{equation}
where $\tilde H_i$ denotes the first $3$ columns of $H_i$.  With this precoding scheme, the 1$^{st}$ and 5$^{th}$ columns are nulled at $\overline Y_1$, while the 2$^{nd}$ and 6$^{th}$ columns are aligned.  In addition, the non-zero columns are linearly independent, so $12$ variables may be solved at $\overline Y_1$.  Similar arguments hold at $\overline Y_2$.  So this scheme can achieve a sum DoF of $24$.

\subsection{Refined interference alignment: $26$ DoF achievable}

A simple refinement of the above scheme leads to even higher DoF. Specifically, zooming into each $H_i$ matrix quickly reveals that it has 1 dimension of null space which we may exploit.  For example, replace each $H_1^\dagger$ and $H_2^\dagger$ in (\ref{eq:5_4}) with $[G_1, \phi_1, \phi_3]$ and $[G_2, \phi_2]$, respectively, where $G_i$ consists of the first $2$ columns of $H_i^\dagger$ and $\phi_i$ is a basis vector of the null space of $H_i$.  Similarly, substitute $[G_3, \phi_3]$ and $[G_4, \phi_4, \phi_2]$ for each $H_3^\dagger$ and $H_4^\dagger$ in (\ref{eq:5_4}), respectively, and the $\mathbf H_\mathrm{eff}$ now becomes:
\begin{equation*} \label{eq:5_6}
\small
\setlength{\dashlinedash}{.4pt}	\setlength{\dashlinegap}{.8pt}
\left[ \arraycolsep=0pt\def\arraystretch{1.1}
\begin{array}{cccc:cccc}
0             	& [\tilde I_3, 0, H_1\phi_3] 	& \tilde H_1 	& 0						& 0	& [\tilde I_3, 0]	& 0			& 0	\\
0			& 0 							& 0 			& 0						& 0	& 0				& 0			& [H_2[G_4, \phi_4],0]	\\
0             	& 0 							& 0 			& H_1[G_3, \phi_3]		& 0	& 0				& 0			& 0	\\
0             	& [\tilde I_3, 0, H_1\phi_3]	& 0 			& 0						& 0	& [\tilde I_3, 0]	& \tilde H_2	& 0	\\
0             	& [\tilde I_3, 0, H_1\phi_3]	& 0 			& H_1[G_3, \phi_3]		& 0	& [\tilde I_3, 0]	& 0			& [H_2[G_4, \phi_4],0]	\vspace{1pt}\\
\hdashline &&&&&&& \\ [-1.0em]
0             	& 0 						& 0	& 0				& 0			& H_4[G_2, \phi_2]	& 0	& 0	\\
\tilde H_3	& 0 						& 0	& [\tilde I_3, 0]	& 0			& 0					& 0	& [\tilde I_3, 0, H_4\phi_2]	\\
0             	& 0 						& 0 	& [\tilde I_3, 0]	& \tilde H_4	& 0					& 0	& [\tilde I_3, 0, H_4\phi_2]	\\
0             	& [H_3[G_1, \phi_1], 0]	& 0 	& 0				& 0			& 0					& 0	& 0	\\
0             	& [H_3[G_1, \phi_1], 0]	& 0 	& [\tilde I_3, 0]	& 0			& H_4[G_2, \phi_2]	& 0	& [\tilde I_3, 0, H_4\phi_2]	\\
\end{array} \right]
\normalsize
\end{equation*}
where $\tilde I_3$ denotes the first $2$ columns of $I_3$.  

Now consider $[G_1, \phi_1, \phi_3]$ first.  The essence is to take away one of the dimensions used by interference alignment, and to save it for interference nulling vectors $\phi_1$ and $\phi_3$.  Since $\phi_1$ vanishes at $\overline Y_1$ and so does $\phi_3$ at $\overline Y_2$, these two vectors occupy only $1$ dimension at either receiver, but they enable us to send one more variable through the network.  The rationale for $[G_4, \phi_4, \phi_2]$ is the same, and the linear independence of the non-zero columns at each receiver is maintained.  

Hence the optimal $26$ DoF is achievable with this scheme.  Moreover, we obtain the code shown in Figure \ref{fig:zf_code} after a slight optimization (of reducing the number of $\phi_i$ filters.)  It is also clear in this view that this code is a space-time code, obtained by interference alignment over space and time (topologies).





\ifdefined\IncludeAppendices

\appendices
\section{DoF upper bound I} \label{App_A}


We prove an upper bound of the sum DoF of the bursty MIMO XC in this appendix, which establishes the converse part of Theorem \ref{thm:DoF} and $\eta_1^\mathrm{ub}$ of Theorem \ref{thm:UBLB}.  To simplify the notations, we define $S\triangleq(S_{11},S_{12}, S_{21}, S_{22})$ and adopt the convention of using $X^n$ to denote a sequence of random variables, $(X_1, X_2, \cdots, X_n)$.
\begin{align} \begin{split} \label{eq:A_1}
&n(R_{11}+R_{12}-\epsilon_n) \\
&\le I(M_{11},M_{12}; Y_1^n,S^n)  \\
&\le I(M_{11},M_{12}; Y_1^n,S^n,M_{21}) \\
&\overset{(a)}{=} I(M_{11},M_{12}; Y_1^n\mid S^n,M_{21}) \\
&= h(Y_1^n\mid S^n, M_{21}) - h(Y_1^n\mid S^n, M_{21},M_{11},M_{12})\\
&\overset{(b)}= h(Y_1^n\mid S^n, M_{21}) - h((S_{12}H_{12}X_2+Z_1)^n\mid S^n, M_{12}),\\
\end{split} \end{align}
where $(a)$ is due to the independence between $(M_{11},M_{12})$ and $(S^n,M_{21})$, $(b)$ follows from the fact that $X_1^n$ becomes deterministic when $S^n, M_{11}, M_{21}$ are given, and $(S_{12}H_{12}X_2+Z_1)^n$ denotes $\{S_{12i}H_{12}X_{2i}+Z_{1i}: i=1,2,\cdots,n\}$.By symmetry, we also have
\begin{align} \begin{split}\label{eq:A_2}
&n(R_{21}+R_{22}-\epsilon_n) \\
&= h(Y_2^n\mid S^n, M_{12}) - h((S_{21}H_{21}X_1+Z_2)^n\mid S^n, M_{21}).\\
\end{split} \end{align}

Let us bound $h(Y_1^n\mid S^n, M_{21})- h((S_{21}H_{21}X_1+Z_2)^n\mid S^n, M_{21})$ first.  Denoting $(S_{21}H_{21}X_1+Z_2)^n$ by $\Omega$, we note that
\begin{align} \begin{split}\label{eq:A_3}
&h(Y_1^n\mid S^n, M_{21})\\
&= I(Y_1^n; X_1^n,X_2^n\mid S^n, M_{21})+h(Z_1^n)\\
&\le I(Y_1^n,\Omega; X_1^n,X_2^n\mid S^n, M_{21})+h(Z_1^n) \\
&= h(Y_1^n, \Omega\mid  S^n, M_{21}) - h(Z_1^n, Z_2^n)+h(Z_1^n)\\
&=h(\Omega\mid  S^n, M_{21})+h(Y_1^n \mid  S^n, M_{21},\Omega)-h(Z_2^n).
\end{split} \end{align}
So it follows that
\begin{align} \begin{split} \label{eq:A_4}
&h(Y_1^n\mid S^n, M_{21})- h((S_{21}H_{21}X_1+Z_2)^n\mid S^n, M_{21})\\
&\le h(Y_1^n \mid  S^n, M_{21},(S_{21}H_{21}X_1+Z_2)^n)\\
&\le h(Y_1^n \mid  S^n,(S_{21}H_{21}X_1+Z_2)^n)\\
&\le \sum_{i=1}^{n} h(Y_{1i}\mid \ S_i, (S_{21}H_{21}X_1+Z_2)_i) \\
&\overset{(a)}\le \sum_{i=1}^{n} h(Y_{1i}\mid \ S_{11i}, S_{12i}, S_{21i}, (S_{21}H_{21}X_1+Z_2)_i),
\end{split} \end{align}
where $(a)$ holds because $S_i=(S_{11i}, S_{12i}, S_{21i}, S_{22i}).$  We then bound $h(Y_{1i}\mid \ S_{11i}, S_{12i}, S_{21i}, (S_{21}H_{21}X_1+Z_2)_i)$ by letting $(S_{11i},S_{12i}, S_{21i})$ assume the values of $(1,1,1), (1,0,1), (0,1,1), (0,0,1), (1,1,0), (1,0,0), (0,1,0),$ and $(0,0,0)$.  This leads to an upper bound of
\begin{equation} \begin{aligned}\label{eq:A_5}
p^3	&h(H_{11}X_{1i}+H_{12}X_{2i}+Z_{1i}&&\mid H_{21}X_{1i}+Z_{2i})\\
+p^2q	&h(H_{11}X_{1i}+Z_{1i}&&\mid H_{21}X_{1i}+Z_{2i})\\
+p^2q	&h(H_{12}X_{2i}+Z_{1i}&&\mid H_{21}X_{1i}+Z_{2i})\\
+pq^2	&h(Z_{1i}&&\mid H_{21}X_{1i}+Z_{2i})\\
+p^2q	&h(H_{11}X_{1i}+H_{12}X_{2i}+Z_{1i}&&\mid Z_{2i})\\
+pq^2	&h(H_{11}X_{1i}+Z_{1i}&&\mid Z_{2i})\\
+pq^2	&h(H_{12}X_{2i}+Z_{1i}&&\mid Z_{2i})\\
+q^3	&h(Z_{1i}&&\mid Z_{2i}).\\
\end{aligned} \end{equation}

Now we distinguish four cases based on $M$ and $N$. In each case, with well-known facts such as Fact 1 of \cite{YW18a}, it is straightforward to verify that if we divide (\ref{eq:A_5}) by $\frac{1}{2}\log P$ and let $P\to \infty$, it is upper bounded by
\begin{enumerate}
\item $N\le M/2:$
\begin{align} \begin{split} \label{eq:A_6}
&N(p^3+p^2q+p^2q+0+p^2q+pq^2+pq^2+0)\\
=&Np(1+q),
\end{split} \end{align}

\item $M/2<N\le M:$
\begin{align} \begin{split} \label{eq:A_7}
&p^3N+p^2q(M-N)+p^2qN+0+p^2qN+pq^2N\\
&+pq^2N+0\\
=&N(p^2+2pq^2)+Mp^2q,
\end{split} \end{align}

\item $M\le N/2:$
\begin{align} \begin{split} \label{eq:A_8}
&M(p^3+0+p^2q+0+2p^2q+pq^2+pq^2+0)\\
=&Mp(1+q),
\end{split} \end{align}

\item$N/2 < M \le N:$
\begin{align} \begin{split} \label{eq:A_9}
&p^3M+0+p^2qM+0+p^2qN+pq^2M+pq^2M+0\\
=&M(p^2+2pq^2)+Np^2q.
\end{split} \end{align}

\end{enumerate}

It is easily checked that (\ref{eq:A_6})--(\ref{eq:A_9}) coincide with Theorem \ref{thm:DoF} and $\eta_1^\mathrm{ub}$ of Theorem \ref{thm:UBLB}, except for a factor $2$. Due to symmetry $h(Y_2^n\mid S^n, M_{12})- h((S_{12}H_{12}X_2+Z_1)^n\mid S^n, M_{12})$ can be bounded in the same way, so the proof is completed by combining (\ref{eq:A_1})--(\ref{eq:A_2}) and (\ref{eq:A_4})--(\ref{eq:A_9}), and letting $n\to \infty.$

\section{DoF upper bound II} \label{App_B}


\begin{figure}[t]
\centering
\includegraphics[width=7cm]{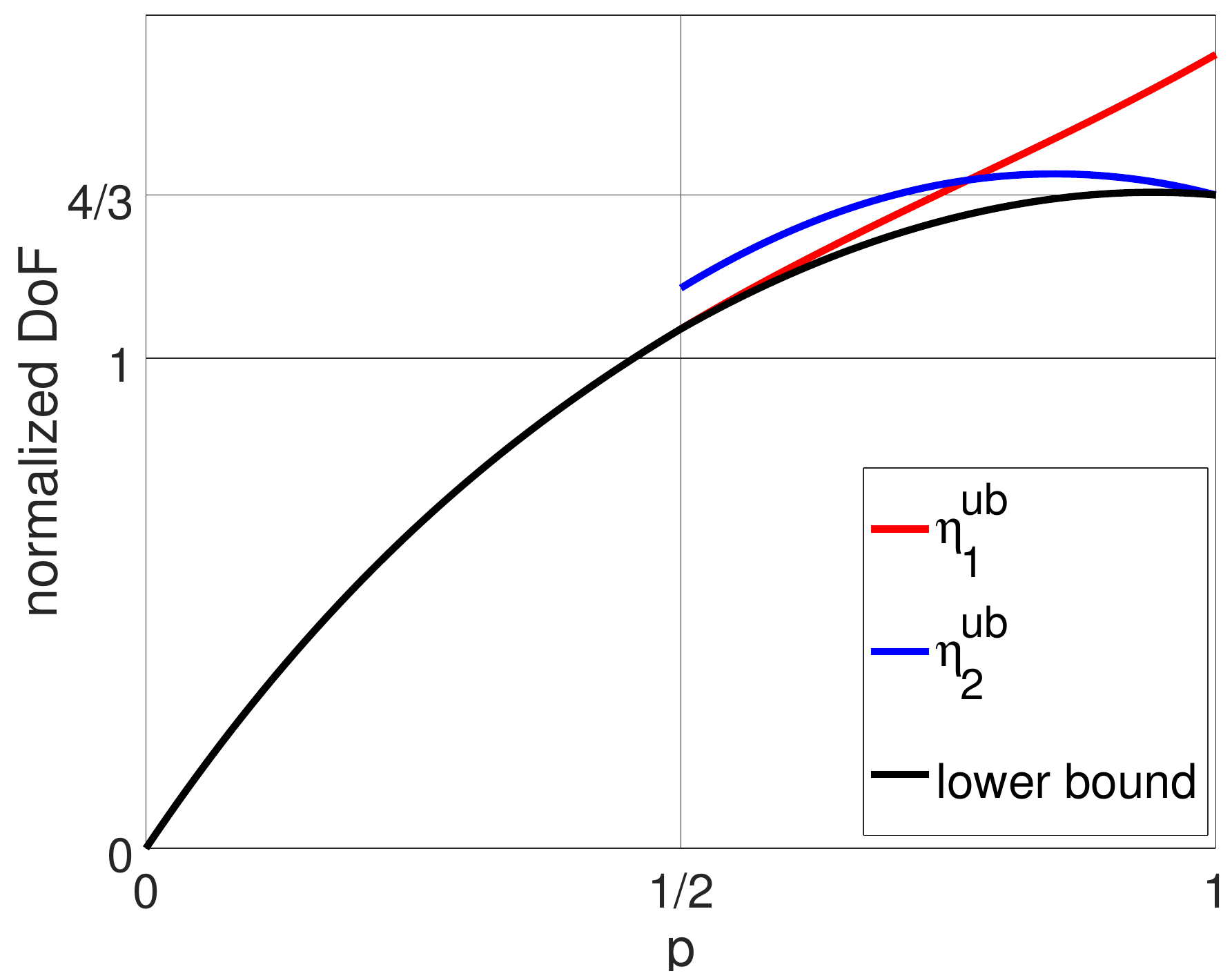}
\caption{$\eta_2^\mathrm{ub}$ complements $\eta_1^{ub}$ when $p$ is large. ($r=0.81$) }
\label{fig:DoF_UB}
\end{figure}

$\eta_2^\mathrm{ub}$ of Theorem \ref{thm:UBLB} is proved by bounding each sum of three rates:
\begin{align} \begin{split} \label{eq:B_1}
&n(R_{21}+R_{22}-\epsilon_{1,n}) \\
&\le I(M_{21},M_{22}; Y_2^n,S^n)  \\
&\le I(M_{21},M_{22}; Y_2^n,S^n,M_{11}) \\
&\overset{(a)}{=} I(M_{21},M_{22}; Y_2^n\mid S^n,M_{11}) \\
&= h(Y_2^n\mid S^n, M_{11}) - h(Y_2^n\mid S^n, M_{11},M_{21},M_{22})\\
&\overset{(b)}= h(Y_2^n\mid S^n, M_{11}) - h((S_{22}H_{22}X_2+Z_2)^n\mid S^n, M_{22}),
\end{split} \end{align}
where the reason for $(a)$ and $(b)$ parallels the one in (\ref{eq:A_1}).  In addition, $R_{12}$ is bounded by
\begin{align} \begin{split} \label{eq:B_2}
&n(R_{12}-\epsilon_{2,n}) \\
&\le I(M_{12}; Y_1^n,S^n)  \\
&\le I(M_{12}; Y_1^n, S^n, Y_2^n, X_1^n, M_{22})  \\
& = I(M_{12}; Y_1^n, Y_2^n\mid S^n, X_1^n, M_{22})  \\
& = h(Y_1^n, Y_2^n\mid S^n, X_1^n, M_{22}) - h(Z_1^n, Z_2^n) \\
& \le h((S_{22}H_{22}X_2+Z_2)^n\mid S^n, M_{22}) \\
& \quad + h((S_{22}H_{22}X_2+Z_2)^n\mid S^n, M_{22}, (S_{12}H_{12}X_1+Z_1)^n).
\end{split} \end{align}

Combining (\ref{eq:B_1}) and (\ref{eq:B_2}), we have
\begin{align} \begin{split} \label{eq:B_3}
&n(R_{12}+R_{21}+R_{22}-\epsilon_{1,n}-\epsilon_{2,n}) \\
& \le  h(Y_2^n\mid S^n) \\
& \quad +  h((S_{22}H_{22}X_2+Z_2)^n\mid S^n, (S_{12}H_{12}X_1+Z_1)^n) \\
& \le \sum_{i=1}^n h(Y_{2i}\mid S_i) \\
& \quad +  h((S_{12}H_{12}X_1+Z_1)_i\mid S_i, (S_{22}H_{22}X_2+Z_2)_i),
\end{split} \end{align}
where
\begin{align} \begin{split} \label{eq:B_4}
& h(Y_{2i}\mid S_i)\\
& \le  p^2h(H_{21}X_{1i}+H_{22}X_{2i}+Z_{2i})+pqh(H_{21}X_{1i}+Z_{2i}) \\
& +pqh(H_{22}X_{2i}+Z_{2i})
\end{split} \end{align}
and 
\begin{align} \begin{split} \label{eq:B_5}
& h((S_{12}H_{12}X_1+Z_1)_i\mid S_i, (S_{22}H_{22}X_2+Z_2)_i) \\
& \le  p^2h(H_{12}X_{1i}+Z_{1i}\mid H_{22}X_{2i}+Z_{2i}) \\
& + pqh(Z_{1i}\mid H_{22}X_{2i}+Z_{2i}) +pqh(H_{12}X_{1i}+Z_{1i}\mid Z_{2i}).
\end{split} \end{align}

Hence, with the same techniques employed in Appendix \ref{App_A}, as $n\to \infty$ and $P\to \infty$, $(R_{12}+R_{21}+R_{22})/(\frac{1}{2}\log P)$ can be upper bounded by
\begin{enumerate}
\item $N\le M:$
\begin{equation} \label{eq:B_6}
p^2N+pqN+pqN+p^2(M-N)+0+pqN=p^2M+3pqN
\end{equation}

\item $N/2\le M\le N:$
\begin{equation} \label{eq:B_7}
p^2N+pqM+pqM+0+0+pqM=p^2N+3pqM,
\end{equation}
\end{enumerate}
which agrees with $\eta_2^\mathrm{ub}$, except for a factor of $\frac{4}{3}$.  The proof is hence completed after a straightforward verification that the same bound applies to $R_{11}+R_{21}+R_{22}, R_{11}+R_{12}+R_{22}$ and $R_{11}+R_{12}+R_{21}$ as well.

Figure \ref{fig:DoF_UB} illustrates how $\eta_2^\mathrm{ub}$ complements $\eta_1^\mathrm{ub}$ when $p$ is large.  (The lower bound is tight when $p\le 0.5.$)

\section{Coding schemes for $M\le N$ and Proof of Theorems \ref{thm:DoF} and \ref{thm:UBLB}} \label{App_C}


\begin{figure}[b]
\centering
\includegraphics[width=9cm]{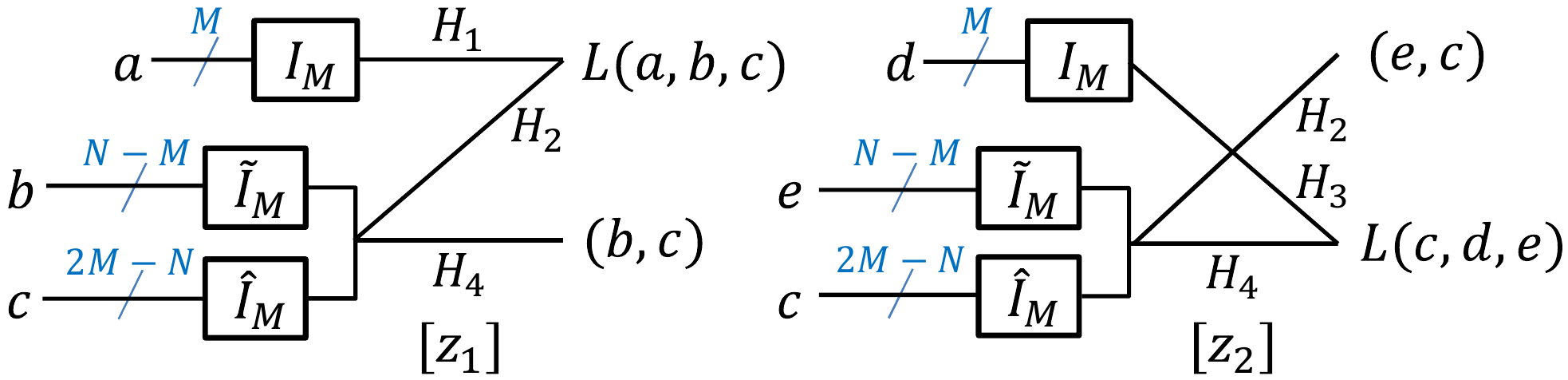}
\caption{Coding scheme for the $\{z_1, z_2\}$ parallel MIMO channel when $M \le N$.  ($\tilde I_M, \hat I_M:$ first $N-M$ and $2M-N$ columns of $I_M$, respectively.)}
\label{fig:z1z2_code_dual}
\end{figure}

The coding schemes for $M\le N$ are similar to those for $M\ge N$, and are included here for completeness.  They are in fact simpler.   Instead of exploiting the null spaces of the channel matrices with interference nulling beamforming (INBF) on the transmitters when $M \ge N$, we just make use of the extra received signal dimensions at the receivers when $M \le N$, eliminating the need for INBF.  We demonstrate this in the proof of the following two dual lemmas.

\begin{lemma} \label{lemma3}
For the $\{z_1, z_2\}$ parallel MIMO channel consisting of the $z_1$ and $z_2$ topologies, $2M+N$ sum DoF is achievable (a.s.), when $\frac{1}{2} < \frac{M}{N} \le 1.$
\end{lemma}
\begin{proof}
Consider the coding scheme illustrated in Figure \ref{fig:z1z2_code_dual}, where $\tilde I_M$ and $\hat I_M$ comprise the first $N-M$ and last $2M-N$ columns of the identity matrix $I_M$, respectively.  Note that we simply send $a$ and $(b,c)$ on Tx$1$ and Tx$2$ of the $z_1$ topology without beamforming (or with the trivial identity beamforming), and $a,b,c$ consist of $M, N-M, 2M-N$ variables, respectively.  Similarly, $d$ and $(e,c)$ are sent on Tx$1$ and Tx$2$ of the $z_2$ topology, respectively.

At high SNR, clearly $(b,c)$ can be decoded reliably at Rx$2$ of the $z_1$ topology, and so can $(e,c)$ at Rx$1$ of the $z_2$ topology.  Hence we may cancel $c$ from the received signal at Rx$1$ of the $z_1$ topology, and solve $(a,b)$ reliably.  By the same token, $(d, e)$ can be retrieved at Rx$2$ of the $z_2$ topology.  This shows that $2M+N$ sum DoF is achievable (a.s.).
\end{proof}

\begin{figure}[t]
\centering
\includegraphics[width=9cm]{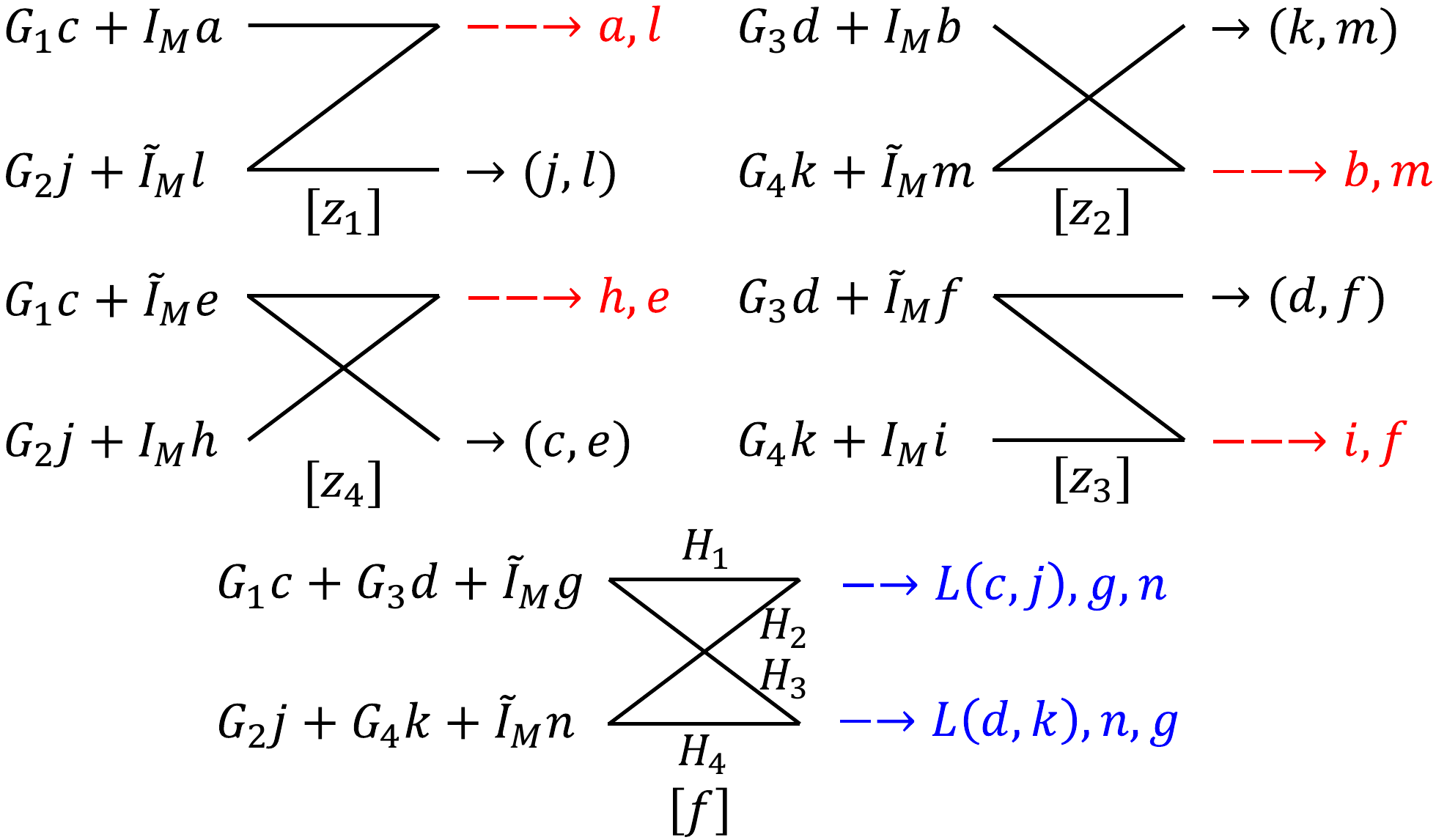}
\caption{Coding scheme for the $\{z,f\}$ parallel MIMO channel when $M \le N$.  ($\tilde I_M$: first $N-M$ columns of $I_M$, $G_i:M\times (2M-N)$ matrices satisfying $H_1G_1=H_2G_2$ and $H_3G_3=H_4G_4$. )}
\label{fig:zf_code_dual}
\end{figure}

\begin{table}[t]
\renewcommand{\arraystretch}{1.5}
\centering
\caption{Length of vectors in the $\{z,f\}$ coding scheme}
\label{tab:zf_vectors_dual}
\vspace{-2ex}
\begin{tabular}{|c|c|c|c|}
	\hline
	  \bfseries vector length & $M$ & $2M-N$ & $N-M$ \\ 
	\hline 
	& & & \\ [-3.5ex]
	\hline
	\bfseries Tx1 vectors & $a,b$ & $c,d$ & $e,f,g$ \\
	\hline
	\bfseries Tx2 vectors & $h,i$ & $j,k$ & $l,m,n$ \\
	\hline
\end{tabular}
\end{table}

\begin{lemma} \label{lemma4}
For the $\{z, f\}$ parallel MIMO channel consisting of the $z_1, z_2, z_3, z_4$ and $f$ topologies, $6M+2N$ sum DoF is achievable (a.s.), when $\frac{2}{3} <\frac{M}{N} \le 1$.
\end{lemma}
\begin{proof}
We use the achievability scheme depicted in Figure \ref{fig:zf_code_dual}, where $\tilde I_M$ again consists of the $N-M$ columns of $I_M$. On the other hand, $G_i$'s are full-rank matrices of dimensions $M\times (2M-N)$ and satisfy $H_1G_1=H_2G_2$ and $H_3G_3=H_4G_4$.  The number of variables contained in $a,b,c,\cdots,n$ vectors are indicated in Table \ref{tab:zf_vectors_dual}.  The decoding closely parallels the 3-step successive interference cancellation procedure in the proof of Lemma \ref{lemma2}:

Step 1 is the same as in the proof of Lemma \ref{lemma2}.

Step 2 is very similar to that of the proof of Lemma \ref{lemma2}, with the only difference being the observation that both $g$ and $n$ can be decoded at either receivers of the $f$ topology due to the extra receiver dimensions afforded by $N\ge M$.

Step 3 is also similar.  The only difference again is that more variables can be decoded at each receiver.  Take Rx$1$ of the $z_1$ topology for example.  $L(c,j)$ is decoded in Step 2 and can be canceled, so only $a$ and $l$ remain, and both can be decoded since we have $N$ antennas at Rx$1$.  Similar arguments hold at the other receivers.

Therefore, $a,b,c,\cdots,n$ can all be reliably solved at high SNR, proving the achievability of $6M+2N$ sum DoF (a.s.).
\end{proof}

With Lemmas \ref{lemma3} and \ref{lemma4}, it is straightforward to complete the proof of the achievability of Theorem \ref{thm:DoF} and $\eta^\mathrm{lb}$ of Theorem \ref{thm:UBLB} for $M\le N$, with obvious arguments dual to those in Section \ref{Achi}.

\fi

%

\IEEEtriggeratref{3}



\end{document}